\documentclass[12pt]{article}
\title{On the Algebraic Structure Underlying the Support Enumerators of Linear Codes}
\author{Nitin Kenjale \& Anuradha S. Garge \\
\small 
K. J. Somaiya Institute of Technology, Sion, \\
\small
Mumbai-400022, India \\ \small \&\\
\small
Department of Mathematics, \\
\small 
University of Mumbai, Mumbai- 400098, India.} 
\date{{\small \today}} 
\usepackage{amsmath,amssymb,amsfonts}
\usepackage{amsthm}
\usepackage{graphicx,latexsym}
\usepackage{tikz} 
\usepackage{amscd, amsmath,amssymb,amsfonts}
\usepackage{amsthm} 
\usepackage{ulem}
\usepackage{graphicx}
\usepackage{amscd, amsmath, amssymb, amsfonts}
\usepackage{amsthm}
\usepackage{ulem}
\usepackage{tabularx}

\newtheorem{thm}{Theorem}[section]
\theoremstyle{definition}

\newtheorem{dfn}{Definition}[section]
\theoremstyle{remark}

\theoremstyle{plain}
\newtheorem{lem}[thm]{Lemma}

\usepackage[pdftex]{hyperref}

\theoremstyle{plain}

\begin{document}
\maketitle{}
\begin{abstract}
  In this paper, we have introduced the concepts of support distribution $(S_i)$ and the support enumerator $S_C(z)$ as refinements of the classical weight distribution and weight enumerator respectively, capturing coordinate level activity in linear block codes. More precisely, we have established formula for counting codewords in the linear code $C$ whose $i^{th}$ coordinate is nonzero. Moreover, we derived a MacWilliams' type identity, relating the normalized support enumerators of a linear code and its dual, explaining how coordinate information transforms under duality. Using this  identity we deduce a new criterion for self-duality based on the equality of support distributions. These results provide a more detailed understanding of code structure and complement classical weight based duality theory.  
\end{abstract}
\vspace{2 mm}

\section{Introduction}
The MacWilliams' identity is one of the central results in coding theory, establishing an elegant relationship between the weight enumerator of a linear block code and that of its dual \cite{RR}. This identity reveals deep algebraic and combinatorial connections between a code and its orthogonal complement, and it has motivated decades of research exploring duality phenomena in linear codes. However, weight enumerators capture only global information, the total number of codewords of a given Hamming weight and therefore do not distinguish how individual coordinates participate in the structure of a code. In many settings, such as local error analysis, coordinate reliability, and the study of coordinate symmetries, this global view is insufficient.

To address this limitation, we define and study another invariant called the support distribution $(S_i)$ and its generating function, the support enumerator $S_C(z)$. We define these in later section. Unlike the classical weight distribution, which collapses all coordinates into a single measure, the support enumerator provides coordinate level information that reflects local structural behavior of the code.

In this paper, we derive a MacWilliams' type identity for support enumerators, relating the normalized support enumerators of a linear code and its dual. The obtained identity captures how coordinate-level information transforms under duality and highlights the role of coordinate positions where the standard basis vectors do not lie in the code or its dual. An important consequence of our identity is a new criterion for self-duality, expressed directly in terms of equality of support distributions.

The paper is organized as follows. Section 1.1 reviews preliminaries on linear codes, additive characters and classical enumerators. Section 2 introduces the support distribution, support enumerator and the algebra underlying them. Section 3 presents the main identity relating support enumerators of a linear code and its dual, together with proofs and structural interpretation. Section 4 discusses necessary condition for self-duality of a linear code. In Section 5 we give illustrative examples and computations for the theorems in preceding sections. 

\subsection{Preliminaries} Let us denote by $\mathbb{F}_q$, the finite field over $q$ elements and by $\mathbb{F}_q^n$, the $n$ dimentional vector space over $\mathbb{F}_q$. We denote by $d(x,y)$, the Hamming distance \cite{RR} between vectors $x,y\in \mathbb{F}_q^n$. We call $C$ a linear code over $\mathbb{F}_q$ if it is a linear subspace of $\mathbb{F}_q^n$ and elements of $C$ are called as codewords of $C$ \cite{Bouyuklieva2024}.

A linear code $C \subseteq \mathbb{F}_q^n$ of dimension $k$ is called an $[n,k,d]$ code, where $n$ is the length of the code (i.e., number of bits in each codeword), $k$ is the dimension (i.e., number of independent message bits) and $d$ is the minimum Hamming distance between any two distinct codewords of $C$.

The fundamental parameters of a linear code include its \textit{weight distribution} and \textit{weight enumerator}, which capture the Hamming weight profile of the code.

The weight distribution of a code $C$ is the sequence $\{A_w\}_{w=0}^n$, where $A_w$ denotes the number of codewords in $C$ of Hamming weight $w$ i.e., 
$A_w = |\{c \in C : \ {wt}(c) = w\}|,$ where $wt(c)$ is the Hamming weight of codeword $c$. 
This information is compactly encoded in the weight enumerator. The (Hamming) weight enumerator of $C$ is the polynomial $\displaystyle W_C(z) = \sum_{w=0}^n A_w z^w.$ Its homogeneous version i.e., homogeneous weight enumerator is 
$\displaystyle
W_C(x, y) = \sum_{c \in C} x^{n - \operatorname{wt}(c)} y^{\operatorname{wt}(c)} = \sum_{w = 0}^n A_w x^{n - w} y^w
,$ (See \cite{FN}).

To relate the weight enumerators of a code and its dual, one typically uses additive characters of the underlying finite field.
Let $\mathbb{F}_q$ be a finite field with $q = p^m$, $p$ prime. An additive character of $\mathbb{F}_q$ is a function
$\chi : \mathbb{F}_q \to \mathbb{C}^*$
such that
$
\chi(a + b) = \chi(a)\chi(b) \  for\ all\ a, b \in \mathbb{F}_q
.$ It follows that, $\chi(0)=1.$ A classical property of additive characters is 
\[
\sum_{x\in\mathbb{F}_q} \chi(x)
=
\begin{cases}
q, & \text{if $\chi$ is the trivial additive character},\\[4pt]
0, & \text{if $\chi$ is a nontrivial additive character}\quad \cite{McE}.
\end{cases}
\]
For a nontrivial additive character $\chi$ and  $c=(c_1,c_2,\dots,c_n)$ one also has
$$
\sum_{\lambda \in\mathbb{F}_q^*} \chi(\lambda c_j)
=
\begin{cases}
q-1, & \text{if } c_j=0,\\[4pt]
-1,  & \text{if } c_j\neq 0 \quad \cite{McE}.
\end{cases}
$$
We also recall the following classical lemma.
\begin{lem} \cite{McE} \label{lem:1.1}
 Let $C\subseteq\mathbb{F}_q^n$ be a linear code and $C^\perp$ be its dual code. Let $\chi:\mathbb{F}_q\to\mathbb{C}^* $ be a nontrivial additive character. For each $u\in\mathbb{F}_q^n,$ define
$
\displaystyle S(u)=\sum_{c\in C}\chi(u\cdot c^t).
$
Then
$$
S(u)=
\begin{cases}
|C|, & u\in C^\perp,\\[4pt]
0, & \text{otherwise}.
\end{cases}
$$
\end{lem}
These properties of additive character are the key analytical tools used to derive the celebrated MacWilliams' identity, that we recall here and which describes how the weight enumerator of a linear code determines that of its dual.
\begin{thm}(MacWilliam's Theorem)\cite{RR}
For every linear code $C$ over $\mathbb{F}_q,$
\[
W_{C^\perp}(x, y) = \frac{1}{|C|} W_C(x + (q - 1)y, x - y).
\]
\qed\end{thm}

\section{The Support Framework}\label{sec:supportframework} 
In this section, we recall the fundamental notion of support of codewords and then we introduce definitions of the support distribution and the support enumerator of a linear code $C\subseteq \mathbb{F}^n_q$.
 
\begin{dfn}{(Support of a Codeword)}\cite{huffman2003} \label{def:3.4}\\
The support of a codeword $c = (c_1, \dots, c_n) \in C$ is defined as
$
\operatorname{supp}(c) = \{i \in \{1, \dots, n\} : c_i \ne 0\}.
$
\end{dfn}
\begin{dfn}{(Support Distribution)}\label{def:3.5}\\
The support distribution of $C$ is a sequence $\{S_i\}_{i=1}^n$, where
$S_i$ counts the number of codewords in $C$ which are nonzero in the $i$-th coordinate i.e., $S_i = |\{c \in C : c_i \ne 0\}|$.
\end{dfn}
\begin{dfn}{(Support Enumerator)}\label{def:3.6}\\
The support enumerator of $C$ is the univariate polynomial 
$
\displaystyle S_C(z)= \sum_{i = 1}^n S_i z^i.
$
This polynomial captures the total nonzero occurrence of the code over each coordinate.
\end{dfn}

Below we establish relation between $S_i$ for $i\in \{1,2,\dots, n\}$ and $A_w$ for $w\in \{0,1,\dots,n\}$.

\begin{thm} \label{thm:4}
Let \( C \subseteq \mathbb{F}_q^n \) be a linear code, then 
$$
\sum_{i=1}^n S_i = \sum_{c \in C} \text{wt}(c) = \sum_{w=1}^n w \cdot A_w.
$$
\end{thm}

\begin{proof}
We begin by considering the total number of nonzero entries in all codewords in \( C \). By definition, for each coordinate \( i \in \{1, \dots, n\} \), the quantity $S_i$ counts the number of codewords \( c \in C \) such that \( c_i \ne 0 \). Therefore,
$$
\sum_{i=1}^n S_i = \sum_{c \in C} \text{wt}(c),
$$
since summing over all coordinates gives the total number of nonzero entries across all codewords.

Now, partition the words of code $C$ weightwise. Let $A_w$ be the number of codewords in \( C \) of weight exactly \( w \). Then the total weight of all the codewords is
$$
\sum_{c \in C} \text{wt}(c) = \sum_{w=1}^n w \cdot A_w,
$$
because each of the $A_w$ codewords of weight \( w \) contribute \( w \) to the total sum.
Combining the above statements, we obtain the desired identity.
$$
\sum_{i=1}^n S_i = \sum_{c \in C} \text{wt}(c) = \sum_{w=1}^n w \cdot A_w.
$$
\end{proof}

\noindent
Now we illustate the support enumerator. We determine the support enumerator of the simplex code. 
Let $ \mathcal{C} $ be the simplex code, which is the dual of the Hamming code over $ \mathbb{F}_q $. Then $\mathcal{C}$  is a linear $[n, m, q^{m-1}]$ code where
$ n = \frac{q^m - 1}{q - 1}$ and
all nonzero codewords in $ \mathcal{C} $ have Hamming weight $ q^{m-1} $ \cite{RR}. Clearly, $ \mathcal{C} $ has $q^m - 1$ nonzero codewords. 
Each codeword has $ q^{m-1} $ nonzero positions \cite{RR},\cite{Mattson}, \cite{Weiss}, so the total number of nonzero entries across all codewords is
$(q^m - 1) \cdot q^{m-1}.$

Since the automorphism group of the simplex code acts transitively on the coordinate positions \cite{Borges2019}, each coordinate appears equally often in the supports of the nonzero codewords. Therefore, the number of codewords in $ \mathcal{C} $ that are nonzero in the $ i $-th coordinate is
$$
S_i = \frac{(q^m - 1) \cdot q^{m-1}}{n} = \frac{(q^m - 1) \cdot q^{m-1}}{\frac{q^m - 1}{q - 1}} = q^{m-1}(q - 1).
$$
Thus, the support enumerator of \( \mathcal{C} \) is given by
$$
S_{\mathcal{C}}(z) = \sum_{i=1}^{n} S_i z^i = q^{m-1}(q - 1) \cdot \sum_{i=1}^{n} z^i.
\hspace{2.5 cm} (*)$$ 

In the following theorem, we derive an expression for the number of codewords in linear code $C$ with $i^{th}$ non-zero coordinate, depending on whether $e_i\in C^\perp$(dual of $C$) or not, where, $e_i$ is the standard basis vector in $\mathbb{F}_q^n$, having $1$ in the $i^{th}$ coordinate and zero elsewhere.   

\begin{thm} \label{thm:3}
Let $C \subseteq \mathbb{F}_q^n$ be a linear code and $C^\perp$ be its dual code. Then
$$
S_i = \begin{cases}
0 & \text{if } e_i \in C^\perp, \\
\frac{(q-1)}{q} |C| & \text{otherwise}.
\end{cases}
$$ 
\end{thm}
\begin{proof}
Let us define two disjoint subsets of $C$ that forms partition of $C$,
\begin{align*}
A &= \{ c \in C : c_i \neq 0 \}\\
B &= \{ c \in C : c_i = 0 \}.
\end{align*}
Clearly, $C = A \sqcup B$, so $|A| = |C| - |B|$. \hfill $(1)$

Our aim is to compute $|B|$, the number of codewords with $0$ in the $i^{th}$ coordinate. We substract it from size of $C$ to get count of codewords with nonzero $i^{th}$ coordiante. Let us define the indicator function as follows:
$$
\delta_0(x) = \begin{cases}
1 & \text{if } x = 0 \\
0 & \text{if } x \neq 0
\end{cases}
$$

We express this function using character sums. Let \( \chi \) be a nontrivial additive character on \( \mathbb{F}_q \). Then for $x=0,$ $\chi(ax)=1$ and $\displaystyle \sum_{a \in \mathbb{F}_q} \chi(ax)=q$. Also, for $x\neq 0,$ $\displaystyle \sum_{a \in \mathbb{F}_q} \chi(ax)=0$, therefore,
$
\displaystyle \delta_0(x) = \dfrac{1}{q} \sum_{a \in \mathbb{F}_q} \chi(ax)
$.
Thus, 
$$
|B| = \sum_{c \in C} \delta_0(c_i) = \sum_{c \in C} \left( \dfrac{1}{q} \sum_{a \in \mathbb{F}_q} \chi(ac_i) \right)
= \dfrac{1}{q} \sum_{a \in \mathbb{F}_q} \sum_{c \in C} \chi(ac_i)
.$$
We split this sum into the terms with $a = 0$ and $a \neq 0$, 
$$
|B| = \frac{1}{q} \left( \sum_{c \in C} 1 + \sum_{a \in \mathbb{F}_q^*} \sum_{c \in C} \chi(ac_i) \right)
= \frac{|C|}{q} + \frac{1}{q} \sum_{a \in \mathbb{F}_q^*} \sum_{c \in C} \chi(ac_i).
$$
We now analyze the inner sum. Fix $a \in \mathbb{F}_q^*$ and observe
$$
\sum_{c \in C} \chi(ac_i) = \sum_{c \in C} \chi(\langle ae_i, c \rangle)
.$$
Let $v = ae_i$. Then, by character orthogonality (See Lemma \ref{lem:1.1}) we get,
$$
\sum_{c \in C} \chi(\langle v, c \rangle)=\sum_{c \in C} \chi(v\cdot c^t) = 
\begin{cases}
|C| & \text{if } v \in C^\perp \\
0 & \text{otherwise.}
\end{cases}
$$
Since $C^\perp$ is also linear, $ae_i \in C^\perp$ if and only if $e_i \in C^\perp$, therefore we get,
$$
\sum_{a \in \mathbb{F}_q^*} \sum_{c \in C} \chi(ac_i) =
\begin{cases}
(q-1)|C| & \text{if } e_i \in C^\perp \\
0 & \text{otherwise.}
\end{cases}
$$
Thus, substituting this in $|B|$ we get,
$$
|B| =
\begin{cases}
\frac{|C|}{q} + \frac{(q-1)|C|}{q} = |C| & \text{if } e_i \in C^\perp \\
\frac{|C|}{q} & \text{otherwise.}
\end{cases}
$$
Substituting $|B|$ in equation $(1)$, we get count of codewords which are nonzero at $i^{th}$ coordinates. 
$$
|A| = |C| - |B| =
\begin{cases}
0 & \text{if } e_i \in C^\perp \\
\frac{(q-1)|C|}{q} & \text{otherwise.}
\end{cases}
$$
\end{proof}

\section{Relating Support Enumerators of $C$ and $C^\perp$}
 
In this section, we establish a \textit{MacWilliams'-type relationship} between the \textit{support enumerators} of a linear code and those of its dual. 
While the classical MacWilliams' identity relates the weight enumerators of a code and its dual via the orthogonality of additive characters, 
here we derive a similar identity by analyzing the \textit{coordinate-level structure} of the code.

This approach exploits the constraints induced by the intersections of the coordinate projections of $C$ and $C^{\perp}$, 
thereby capturing how the supports of codewords are distributed across coordinates in $\mathbb{F}_q^n$. 
The resulting identity provides a direct algebraic connection between $S_C(z)$ and $S_{C^{\perp}}(z)$, 
incorporating the contributions of coordinates that belong to $C$, $C^{\perp}$, and their intersection.

\begin{thm} \label{thm:4}
 Let $C \subseteq \mathbb{F}_q^n$ be a linear code and let $C^\perp$ be its dual code. Let $S_C(z)$ and $S_{C^{\perp}}(z)$ be support enumerators of code $C$ and $C^\perp$ respectively. 
Then the following identity holds:
$$
\frac{1}{|C|} S_C(z) + \frac{1}{|C^\perp|} S_{C^\perp}(z)
= \frac{q - 1}{q} \left( \sum_{i = 1}^n z^i + \sum_{\{i : e_i \notin C \cup C^\perp\}} z^i \right).
$$
\end{thm}
\begin{proof}
Let $e_i \in \mathbb{F}_q^n$ denote the standard basis vector.
By Theorem $\ref{thm:3}$, 
$$
S_i= |\{c \in C : c_i \ne 0\}|
=
\begin{cases}
0 & \text{if } e_i \in C^\perp \\
\frac{q - 1}{q} \cdot |C| & \text{otherwise.} 
\end{cases}
$$
Similarly, $S_i^\perp$ denoting the number of codewords in $C^\perp$ with $i^{th}$ nonzero coordinate is given by 
$$
S_i^\perp =
|\{c \in C^\perp : c_i \ne 0\}|
=
\begin{cases}
0 & \text{if } e_i \in C \\
\frac{q - 1}{q} \cdot |C^\perp| & \text{otherwise}.
\end{cases}
$$
Normalizing $S_i$ and $S_i^\perp$ by dividing their code sizes we get,   
$$
\frac{S_i}{|C|} =
\begin{cases}
0 & \text{if } e_i \in C^\perp \\
\frac{q - 1}{q} & \text{otherwise}
\end{cases},
\quad
\frac{S_i^\perp}{|C^\perp|} =
\begin{cases}
0 & \text{if } e_i \in C \\
\frac{q - 1}{q} & \text{otherwise.}
\end{cases}
$$
Now partition the coordinates $\{1, 2, \dots, n\}$ into four disjoint subsets,
 $$A = \{i : e_i \in C \cap C^\perp\}
,\quad B = \{i : e_i \in C \setminus C^\perp\}$$
$$C = \{i : e_i \in C^\perp \setminus C\}, 
\quad D = \{i : e_i \notin C \cup C^\perp\}$$
We now analyze values of $S_i$ and $S^\perp_i$ by taking cases on above partitions.
\begin{enumerate}
 \item For $i \in A$, $e_i$ lies in $C$ and in $C^\perp.$ This is an imppossible case, as $(e_i\cdot e_i^t)\neq 0.$ 
 \item For $i \in B \cup C$, exactly one of $S_i$, $S_i^\perp$ is nonzero and so it contributes $\frac{q - 1}{q} z^i$ in the enumerator.
 \item For $i \in D$, both terms are nonzero and hence such an $i$  contributes $2 (\frac{q - 1}{q} z^i)$.
\end{enumerate}
We now compute,
\begin{align*}
  \frac{1}{|C|} S_C(z) + \frac{1}{|C^\perp|} S_{C^\perp}(z)
&= \sum_{i=1}^n \left( \frac{S_i}{|C|} + \frac{S_i^\perp}{|C^\perp|} \right) z^i\\
&= \frac{q - 1}{q} \left( \sum_{i \in B \cup C} z^i + 2 \sum_{i \in D} z^i \right).
 \end{align*}
Now rewrite the expression using following
$$
\sum_{i = 1}^n z^i = \sum_{i \in B \cup C \cup D} z^i,\  i.e.,
\sum_{i \in B \cup C} z^i +  \sum_{i \in D} z^i = \sum_{i = 1}^n z^i.
$$
Hence we have
$$
\frac{1}{|C|} S_C(z) + \frac{1}{|C^\perp|} S_{C^\perp}(z)
= \frac{q - 1}{q} \left( \sum_{i = 1}^n z^i + \sum_{i \in D} z^i \right).
$$
That is, 
$$
\frac{1}{|C|} S_C(z) + \frac{1}{|C^\perp|} S_{C^\perp}(z)
= \frac{q - 1}{q} \left( \sum_{i = 1}^n z^i + \sum_{\{i : e_i \notin C \cup C^\perp\}} z^i\right).
$$
\end{proof}
\noindent Henceforth, we refer to Theorem \ref{thm:4} as the support enumerator identity. It is straightforward to express support enumerator of $C^\perp$ in terms of support enumerator of $C$ and conversely. 

\section{Necessary condition for Self-Duality}
In this section, we will determine necessary condition for a code to be self-dual. We first recall definition of self-dual codes and then using support enumerator identity i.e., Theorem \ref{thm:4} we derive new criterion to identify whether code $C$ is self-dual or not.  

\begin{dfn}{ (Self-Dual Code)} \cite{RR}\\
A code is called self-dual if it is equal to its dual i.e.,  $C=C^\bot.$
\end{dfn}

For example,  the linear $[8,4,4]$ extended Hamming code over $\mathbb{F}_2$ is self-dual, as it can be generated from its parity check matrix. 

\begin{thm}\textbf{(Necessary condition for Self-duality)} \label{thm:5}\\
Let \( C \subseteq \mathbb{F}_q^n \) be a linear code. If $C$ is self-dual then
$$
\frac{1}{|C|} S_C(z) = \frac{q-1}{2q} \left( \sum_{i=1}^n z^i + \sum_{\{i:e_i \notin C\}} z^i \right).
$$ 
\end{thm}
\begin{proof}
If a linear code $C$ is self-dual, i.e., $C = C^\perp$, then it is easy to see that $C \cup C^\perp = C$ and $C \cap C^\perp = C.$ Also, $|C| = |C^\perp|$ and $S_C(z) = S_{C^\perp}(z)$ holds.

We substitute these statements in Theorem \ref{thm:4}, which states
$$
\frac{1}{|C|} S_C(z) + \frac{1}{|C^\perp|} S_{C^\perp}(z) = \frac{q-1}{q} \left( \sum_{i=1}^n z^i + \sum_{\{i:e_i \notin C \cup C^\perp\}} z^i\right).
$$
Therefore we get
$$
\frac{2}{|C|} S_C(z) = \frac{q-1}{q} \left( \sum_{i=1}^n z^i + \sum_{\{i:e_i \notin C\}} z^i \right).  
$$
Thus, we get,
$$
\frac{1}{|C|} S_C(z) = \frac{q-1}{2q} \left( \sum_{i=1}^n z^i + \sum_{\{i:e_i \notin C\}} z^i \right).
$$

\end{proof}

\section{Illustrations of Theorems}
In this section, we provide explicit examples illustrating the theorems established in the earlier sections. These examples are intended to demonstrate the applicability of the results and to clarify the underlying constructions.

\subsection{Illustration of support enumerator identity}

Below we illustrate support enumerator identity i.e., Theorem \ref{thm:4} using simplex and repetition codes. 
\begin{enumerate}
\item Illustration using simplex code over $\mathbb{F}_q$:

Let \( C \subseteq \mathbb{F}_q^n \) be the \([n, m, q^{m - 1}]\) simplex code, where \( n = \frac{q^m - 1}{q - 1} \), and let \( C^\perp \subseteq \mathbb{F}_q^n \) be its dual Hamming code of parameters \([n, n - m, 3]\).

\noindent First we find support enumerator of the Simplex Code $C$.
From $(*)$ of Section~\ref{sec:supportframework}, we see,  support enumerator of Simplex code is
$$
S_{\mathcal{C}}(z) = \sum_{i=1}^{n} S_i z^i = q^{m-1}(q - 1) \cdot \sum_{i=1}^{n} z^i.
$$ 

Hence,
$$
\frac{1}{|C|} S_C(z) = \frac{1}{q^m} \sum_{i = 1}^n S_i z^i = \frac{1}{q^m} \cdot q^{m - 1}(q - 1) \sum_{i = 1}^n z^i = \frac{q - 1}{q} \sum_{i = 1}^n z^i.
$$

\noindent Next we find support enumerator of the dual of Simplex code, $C^\perp$ i.e., Hamming Code. 
From \cite{FN},\cite{huffman2003},\cite{Borges2019}, we see that in $[n, n - m, 3]$ Hamming code, the total number of codewords is $|C^\perp| = q^{n - m}$ and each nonzero coordinate appears in exactly \( q^{n - m - 1}(q - 1) \) codewords. 

Hence,
$$
\frac{1}{|C^\perp|} S_{C^\perp}(z) = \frac{q^{n - m - 1}(q - 1)}{q^{n - m}} \sum_{i = 1}^n z^i = \frac{q - 1}{q} \sum_{i = 1}^n z^i
$$
Adding both enumerators we get left-hand side of the support enumerator identity,
$$
\frac{1}{|C|} S_C(z) + \frac{1}{|C^\perp|} S_{C^\perp}(z)
= \frac{q - 1}{q} \sum_{i = 1}^n z^i + \frac{q - 1}{q} \sum_{i = 1}^n z^i
= \frac{2(q - 1)}{q} \sum_{i = 1}^n z^i.
$$
\noindent For right-hand side of the identity we examine and compute following terms in the identity,
$$
\sum_{i = 1}^n z^i + \sum_{\{i : e_i \notin C \cup C^\perp\}} z^i.
$$

In this setting, no standard basis vector $e_i$ is in the Simplex code $C$, because all nonzero codewords in $C$ have full support when $q>2$ and when $q=2$ (binary), nontrivial Simplex code has minimum hamming weight greater than or equal to $2$. 

Also, no $e_i \in C^\perp$ since the Hamming code has minimum distance $\geq 3$. Therefore, $e_i \notin C \cup C^\perp$ for all $ i$.

So,
$
\displaystyle \sum_{\{i : e_i \notin C \cup C^\perp\}} z^i = \sum_{i = 1}^n z^i.
$

Thus, the right-hand side becomes:
$$
\frac{q - 1}{q} \left( \sum_{i = 1}^n z^i + \sum_{i = 1}^n z^i \right) = \frac{2(q - 1)}{q} \sum_{i = 1}^n z^i,
$$
which is identical to the left-hand side. Hence, the support enumerator identity is verified for the Simplex-Hamming dual pair.

\item {Illustration using simplex code $[7,3,4]$ over $\mathbb{F}_2$:}

Let \( C \subseteq \mathbb{F}_2^7 \) be the binary simplex code with parameters \( [7,3,4] \), and let \( C^\perp \) be its dual Hamming code \( [7,4,3] \). Then,
$
|C| = 2^3 = 8, \quad |C^\perp| = 2^4 = 16.
$\\
Let a generator matrix for $C$ be $G$ and generator matrix of it's dual be $H$ given by
$$G = \begin{bmatrix}
1 & 0 & 0 & 1 & 1 & 0 & 1 \\
0 & 1 & 0 & 1 & 0 & 1 & 1 \\
0 & 0 & 1 & 0 & 1 & 1 & 1
\end{bmatrix}, 
H = \begin{bmatrix}
1 & 0 & 0 & 0 & 1 & 1 & 0 \\
0 & 1 & 0 & 0 & 1 & 0 & 1 \\
0 & 0 & 1 & 0 & 0 & 1 & 1 \\
0 & 0 & 0 & 1 & 1 & 1 & 1
\end{bmatrix}.$$
Below is the list of codewords of $C$ and $C^{\perp}$ respectively.

$${
\begin{array}{c|c}
\text{Message} & \text{Codeword} \\
\hline
000 & 0000000 \\
001 & 0010111 \\
010 & 0101011 \\
011 & 0111100 \\
100 & 1001101 \\
101 & 1011010 \\
110 & 1100110 \\
111 & 1110001 
\end{array}
\begin{array}{c|c}
\text{Message} & \text{Codeword} \\
\hline
0000 & 0000000 \\
0001 & 0001111 \\
0010 & 0010011 \\
0011 & 0011100 \\
0100 & 0100101 \\
0101 & 0101010 \\
0110 & 0110110 \\
0111 & 0111001 \\
1000 & 1000110 \\
1001 & 1001001 \\
1010 & 1010101 \\
1011 & 1011010 \\
1100 & 1100011 \\
1101 & 1101100 \\
1110 & 1110000 \\
1111 & 1111111 \\
\end{array}
}$$

Now we find the support enumerator for the Simplex code and the Hamming code.
For $C$ we see, $S_i=4$ for all $1\leq i\leq7$ and for $C^\perp$, $S_j=8$ for all $1\leq j\leq7$.\\
Therefore, support enumerators for $C$ and $C^\perp$ will be respectively,
$$
S_C(z) = 4z^1 + 4z^2 + 4z^3 + 4z^4 + 4z^5 + 4z^6 + 4z^7,
$$
$$
S_C^\perp(z) = 8z^1 + 8z^2 + 8z^3 + 8z^4 + 8z^5 + 8z^6 + 8z^7.
$$
\noindent Thus,\\ 
$\frac{1}{|C|} S_C(z) + \frac{1}{|C^\perp|} S_{C^\perp}(z)=\frac{1}{8} S_C(z) + \frac{1}{16} S_{C^\perp}(z)=z^1 + z^2 + z^3 + z^4 + z^5 + z^6 + z^7=\sum_{i=1}^7 z^i.$ 

For right-hand side of the identity, as \( q = 2 \), we have \( \frac{q-1}{q} = \frac{1}{2} \), so right-hand side becomes 
$$
= \frac{1}{2} \left( \sum_{i=1}^7 z^i + \sum_{\{i:e_i \notin C \cup C^\perp\}} z^i \right)
$$

We analyze that $e_i \notin C\cup C^\perp$ for all $i$.  Hence right-hand side equals
$$
 = \frac{1}{2} \left( \sum_{i=1}^7 z^i + \sum_{i=1}^7 z^i \right)=\sum_{i=1}^7 z^i.
$$
This is identical to left-hand side. Thus, 
support enumerator identity is verified i,e.,
$$
{
\frac{1}{|C|} S_C(z) + \frac{1}{|C^\perp|} S_{C^\perp}(z)
= \frac{q - 1}{q} \left( \sum_{i = 1}^n z^i + \sum_{\{i : e_i \notin C \cup C^\perp\}} z^i \right)
}
$$
holds true for $[7,3,4]_2$ simplex code.

\item Illustration using Repetition code $[3,1,3]$ over ${\mathbb{F}_2}$ (binary).

Let \( C \subseteq \mathbb{F}_2^3 \) be the binary repetition code of length 3, given by:
$$
C = \{000, 111\}, \quad |C| = 2
$$
Its dual, is the even-weight code,
$
C^\perp = \{000, 011, 101, 110\}, \ |C^\perp| = 4.
$

\noindent Now we find support enumerators for $C$ and $C^\perp$. The only nonzero codeword in $C$ is $(111)$, which has support at all three coordinates. So:
$$
S_C(z) = z^1 + z^2 + z^3, \quad \text{and} \ \frac{1}{|C|} S_C(z) = \frac{1}{2}(z^1 + z^2 + z^3).
$$

In $C^\perp,$ each coordinate appears in exactly 2 of the 3 nonzero codewords, so:
$$
S_{C^\perp}(z) = 2z^1 + 2z^2 + 2z^3, \quad \frac{1}{|C^\perp|} S_{C^\perp}(z) = \frac{1}{4}(2z^1 + 2z^2 + 2z^3) = \frac{1}{2}(z^1 + z^2 + z^3).
$$
Hence,$$
\frac{1}{|C|} S_C(z) + \frac{1}{|C^\perp|} S_{C^\perp}(z) = \frac{1}{2}(z^1 + z^2 + z^3) + \frac{1}{2}(z^1 + z^2 + z^3) = z^1 + z^2 + z^3.
$$

For right-hand side we see, 
$
\frac{q - 1}{q} = \frac{1}{2},
$

We check the standard basis vectors $e_1 = 100, e_2 = 010, e_3 = 001$, none of these belong to $C$ or $C^\perp$, so:
$$
\sum_{\{i : e_i \notin C \cup C^\perp\}} z^i = \sum_{i=1}^3 z^i.
$$

Hence, right-hand side of the support enumerator identity becomes 
\begin{align*}
\frac{q - 1}{q} \left( \sum_{i = 1}^n z^i + \sum_{\{i : e_i \notin C \cup C^\perp\}} z^i  \right)&=\frac{1}{2} \left( \sum_{i=1}^3 z^i + \sum_{i=1}^3 z^i\right)\\ &= z^1 + z^2 + z^3.
\end{align*}

This proves the support enumerator identity, $$
\frac{1}{|C|} S_C(z) + \frac{1}{|C^\perp|} S_{C^\perp}(z) = \frac{q - 1}{q} \left( \sum_{i = 1}^n z^i + \sum_{\{i : e_i \notin C \cup C^\perp\}} z^i \right)
.$$
\end{enumerate}

\subsection{Illustration of Necessary condition for self-duality}
Here we illustrate criterion for self-duality i.e., Theorem \ref{thm:5} using the linear $[8,4,4]$ extended Hamming code $C$ over $\mathbb{F}_2$. 

Code $C$ has
$
n = 8, \  k = 4, \  q = 2, \  |C| = 2^{k} = 2^4 = 16
.$
Below is the list of codewords in $C$.\\
$$\begin{array}{cccc}
00000000 & 11110000 & 11001100 & 10101010 \\
00111100 & 01011010 & 10010110 & 01100110 \\
10011001 & 01010101 & 11000011 & 01111000 \\
10110100 & 00011110 & 00100011 & 11101111 \\
\end{array}$$

\noindent We now compute $S_i$ for $1\leq i\leq 8$. Since each coordinate appears with a nonzero value in exactly $8$ codewords, so $S_i=8$ for all $i$. Therefore,
$$
S_C(z) = \sum_{i=1}^8 S_i z^i = 8z + 8z^2 + 8z^3 + 8z^4 + 8z^5 + 8z^6 + 8z^7 + 8z^8 = 8 \sum_{i=1}^8 z^i.
$$
Now divide by $|C| = 16$,
$\displaystyle
\frac{1}{|C|} S_C(z) = \frac{8}{16} \sum_{i=1}^8 z^i = \frac{1}{2} \sum_{i=1}^8 z^i.
$
\noindent
This gives the left-hand side of the criterion.

\noindent Now we compute the right-hand side i.e., 
$
\displaystyle \frac{q-1}{2q} \left( \sum_{i=1}^n z^i + \sum_{\{i:e_i \notin C\}} z^i \right).
$\\
Each standard basis vector \( e_i \) has weight 1. Since all nonzero codewords in \( C \) have weight at least 4, none of the \( e_i \) are in \( C \). Hence 
$ \displaystyle
\sum_{\{i:e_i \notin C\}} z^i = \sum_{i=1}^8 z^i
$.
Thus, the right-hand side becomes  
$ \displaystyle
\frac{1}{4} \left(\sum_{i=1}^8 z^i
+\sum_{i=1}^8 z^i\right)$ i.e., 
$ \displaystyle
\frac{1}{2} \left(\sum_{i=1}^8 z^i
\right),$
which is same as the left-hand side. This verifies necessary condition for self-duality.

\end{document}